\documentclass[submission]{eptcs}
\usepackage{amsmath}
\usepackage{amssymb}
\usepackage{amsthm}
\usepackage{bussproofs}
\usepackage{graphicx}
\usepackage{stmaryrd}
\usepackage[all,cmtip]{xy}
\usepackage{verbatim}
\usepackage{xspace}

\title{QPEL: Quantum Program and Effect Language}
\def\titlerunning{QPEL}
\author{Robin Adams
\institute{Radboud University Nijmegen}
\email{r.adams@cs.ru.nl}}
\def\authorrunning{Robin Adams}
\providecommand{\event}{QPL 2014}

\newtheorem{theorem}{Theorem}
\newtheorem{lemma}[theorem]{Lemma}
\newtheorem{lm}[theorem]{Lemma}
\newtheorem{corollary}{Corollary}[theorem]
\theoremstyle{definition}
\newtheorem{df}[theorem]{Definition}

\newcommand{\pareq}{\stackrel{\sim}{\rightarrow}}
\newcommand{\Set}{\mathbf{Set}}
\newcommand{\supp}{\operatorname{supp}}

\newcommand{\prop}{\ \mathrm{eff}}
\newcommand{\brackets}[1]{[\![ {#1} ]\!]}
\newcommand{\EMod}[1]{\mathbf{EMod}_{#1}}
\newcommand{\op}{\mathrm{op}}
\newcommand{\Conv}[1]{\mathbf{Conv}_{#1}}
\newcommand{\inl}[1]{\mathsf{inl} \left( {#1} \right)}
\newcommand{\inr}[1]{\mathsf{inr} \left( {#1} \right)}
\newcommand{\qbit}{\mathbf{qbit}}
\newcommand{\new}{\operatorname{new}}
\newcommand{\ket}[1]{\left| {#1} \right\rangle}

\newcommand{\FdHilbUn}{\mathbf{FdHilb}_\mathrm{Un}}
\newcommand{\Kl}[1]{\mathrm{Kl} \left( {#1} \right)}
\newcommand{\CStar}{\mathbf{CStar}_\mathrm{PU}^\mathrm{op}}

\newcommand{\trace}{\operatorname{tr}}

\begin{document}

\maketitle

\begin{abstract}
We present the syntax and rules of deduction of QPEL (Quantum Program and Effect Language), a language for describing both quantum programs, and properties of quantum programs --- \emph{effects} on the appropriate Hilbert space.  We show how semantics may be given in terms of \emph{state-and-effect triangles}, a categorical setting that allows semantics in terms of Hilbert spaces, C$^*$-algebras, and other categories.  We prove soundness and completeness results that show the derivable judgements are exactly those provable in all state-and-effect triangles.
\end{abstract}

\section{Introduction}

There is a growing number of quantum programming languages, and there is a need for a syntactic method of reasoning about these quantum programs: both in the hope of making automated tools for proving the correctness of programs, and because experience in other fields shows that many problems that are difficult when treated semantically

We present QPEL, a syntax for both describing quantum programs, and properties of quantum programs (quantum predicates, or effects).  This system should be useful for reasoning about quantum programs and proving their correctness, as well as showing more generally how a language for quantum effects may be added on top of any quantum programming language.
The part of the system that descibes quantum programs is loosely based on Selinger's Quantum Programming Language (QPL) \cite{Selinger2004}.

The part of the system that describes quantum programs is a \emph{linear} type theory (see \cite{Benton93aterm,Benton1995}): we are not able to duplicate data.  Duplication of quantum data would violate the no-cloning theorem.  We do allow deletion of data (which corresponds to e.g. measuring a qubit then throwing away the outcome of the measurement).

The part of the system that describes quantum predicates is based on the fact that the effects on a Hilbert state or C$^*$-algebra form an \emph{effect algebra} - in fact, an \emph{effect module} over the appropriate effect monoid \cite{Jacobs2013}.  

There is a categorical structure called the \emph{state-and-effect triangle} that has been shown to generalise several different ways of giving semantics to quantum computing, including Hilbert spaces and C$^*$-algebras.  The first version of QPEL we present captures all and only the structure of a state-and-effect triangle.  We show how to give semantics in an arbitrary triangle, and prove a Soundness and Completeness Theorem.  We proceed to discuss what would need to be added to the system to represent other features of a quantum programming language, particularly qubits.

The language QPEL has a homepage at \texttt{www.cs.ru.nl/$\sim$robina/QPEL}

\pagebreak

\section{Preliminaries}

\subsection{Notation}
If $E$ and $F$ are expressions involving partial functions, we write:
\begin{itemize}
\item
$E = F$
to denote: $E$ and $F$ are both defined, and their values are equal;
\item
$E \simeq F$
to denote: $E$ is defined if and only if $F$ is defined, in which case their values are equal (this is sometimes known as \emph{Kleene equality});
\item
$E \pareq F$
to denote: if $E$ is defined, then $F$ is defined and their values are equal (this is sometimes known as \emph{directed equality}).
\end{itemize}

\subsection{Effect Algebras and Effect Monoids}

We represent the effects on a quantum system by the elements of an \emph{effect module}
over an \emph{effect monoid} $M$, whose elements we call \emph{scalars}.  The canonical example is the effects on a Hilbert space or C$^*$-algebra, which form an effect module over $[0,1]$, with the scalars being probabilities.  These concepts were introduced in \cite{Jacobs2011}.

\begin{df}[Partial Commutative Monoid]
A \emph{partial commutative monoid} consists of a set $M$; an element $0 \in M$, the \emph{zero}; and a partial binary operation $\ovee : M^2 \rightharpoonup M$, the \emph{(partial) sum}; such that:
\begin{itemize}
\item
$x \ovee y \simeq y \ovee x$
\item
$x \ovee (y \ovee z) \simeq (x \ovee y) \ovee z$
\item
$x \ovee 0 = x$
\end{itemize}
for all $x,y,z \in M$.

We write $x \perp y$, $x$ is \emph{orthogonal} to $y$, iff $x \ovee y$ is defined.
\end{df}

\begin{df}[Effect Algebra]
An \emph{effect algebra} is a partial commutative monoid $E$ with a (total) function $(-)^\bot : E \rightarrow E$, the \emph{orthosupplement}, such that
\begin{itemize}
\item
$x \ovee y = 0^\bot$ iff $y = x^\bot$.
\item
If $x \perp 0^\bot$ then $x = 0$.
\end{itemize}

We write $1$ for $0^\bot$.
\end{df}

\begin{df}[Effect Algebra Homomorphism]
Let $E$ and $F$ be effect algebras.  An \emph{effect algebra homomorphism} $\phi : E \rightarrow F$ is a function such that, for all $x,y \in E$:
\begin{align*}
\phi(x \ovee y) & \pareq \phi(x) \ovee \phi(y) \\
\phi(x^\bot) & = \phi(x)^\bot
\end{align*}
\end{df}

\begin{lm}
For any effect algebra homomorphism $\phi$, we have $\phi(0) = 0$.
\end{lm}

\begin{proof}
\begin{align*}
\phi(0 \ovee 0) & = \phi(0) \\
\therefore \phi(0) \ovee \phi(0) & = \phi(0) \\
& = \phi(0) \ovee 0 \\
\therefore \phi(0) & = 0
\end{align*}
since any effect algebra satisfies cancellation \cite{Dvurecenskij2000}.
\end{proof}

\begin{df}[Effect Monoid]
An \emph{effect monoid} is an effect algebra $E$ with a binary operation $\cdot : E^2 \rightarrow E$, the \emph{multiplication}, such that
\begin{itemize}
\item
$(x \ovee y) \cdot z \pareq (x \cdot z) \ovee (y \cdot z)$
\item
$x \cdot (y \ovee z) \pareq (x \cdot y) \ovee (x \cdot z)$
\item
$1 \cdot x = x \cdot 1 = x$
\item
$x \cdot (y \cdot z) = (x \cdot y) \cdot z$
\end{itemize}
The effect monoid is \emph{commutative} iff $x \cdot y = y \cdot x$ for all $x$, $y$.
\end{df}

An effect monoid is a monoid in the category of effect algebras \cite{Jacobs2011}.

\begin{lm}
In any effect monoid, $x \cdot 0 = 0 \cdot x = 0$.
\end{lm}

\begin{proof}
We have
\begin{align*}
x \cdot (0 \ovee 0) & = x \cdot 0 \\
\therefore x \cdot 0 \ovee x \cdot 0 & = x \cdot 0 \\
& = x \cdot 0 \ovee 0 \\
\therefore x \cdot 0 & = 0
\end{align*}
by cancellation.  Similarly for $0 \cdot x$.
\end{proof}

\begin{df}[Effect Module]
An \emph{effect module} over an effect monoid $E$ is an effect algebra $A$ with a binary operation $\cdot : E \times A \rightarrow A$
called \emph{scalar multiplication} such that, for all $x, y, z \in E$:
\begin{itemize}
\item
$r \cdot (x \ovee y) \pareq (r \cdot x) \ovee (r \cdot y)$
\item
$(r \ovee s) \cdot x \pareq (r \cdot x) \ovee (s \cdot x)$
\item
$(r \cdot s) \cdot x = r \cdot (s \cdot x)$
\item
$1 \cdot x = x$
\end{itemize}
\end{df}

\begin{df}[Effect Module Homomorphism]
Let $A$ and $B$ be effect modules over $E$.  An \emph{effect module homomorphism} $\phi : A \rightarrow B$ is an effect algebra homomorphism such that, for all $r \in E$ and $x \in A$,
\[ \phi(r \cdot x) = r \cdot \phi(x) \enspace . \]
\end{df}

\subsubsection{Examples}

\begin{enumerate}
\item
For any Hilbert space $H$, the set of effects over $H$ forms an effect module over the effect monoid $[0,1]$, with $F \ovee G = F + G$ iff $F + G$ is an effect \cite{Jacobs2013}.
\item
Given a C$^*$-algebra $A$, the set of \emph{effects} in $A$ (positive elements below the unit) form an effect module over the real numbers $[0,1]$.
\end{enumerate}

\subsection{Convex Sets}

We describe the category of \emph{convex sets} over any effect monoid.  The \emph{states} of a quantum system will form a convex set over the effect monoid of probabilities.

\begin{df}
Given an effect monoid $E$, the \emph{distribution monad} $\mathcal{D}_E : \Set \rightarrow \Set$ is defined as follows.

\[ \mathcal{D}_E X = \{ \phi : X \rightarrow E : \supp \phi \mbox{ is finite}, \sum_{x \in X} \phi(x) \mbox{ exists and is equal to } 1 \} \]
where $\supp \phi = \{ x \in X : \phi(x) \neq 0 \}$.

For $f : X \rightarrow Y$,
\[ \mathcal{D}_E f(\phi)(y) = \sum_{f(x) = y} \phi(x) \qquad (\phi \in \mathcal{D}_E X, y \in Y) \enspace . \]

The unit $\eta_A : A \rightarrow \mathcal{D}_E A$ is defined by
\[ \eta_A(a)(a') = \begin{cases}
1 & \mbox{if } a = a' \\
0 & \mbox{if } a \neq a'
\end{cases} \]

The multiplication $\mu_A : \mathcal{D}_E^2 A \rightarrow \mathcal{D}_E A$ is defined by
\[ \mu_A(\Phi)(a) = \sum_{\phi \in \mathcal{D}_E A} (\Phi(\phi) \cdot \phi(a)) \enspace . \]
\end{df}

The category $\Conv{E}$ of \emph{convex sets} and \emph{affine functions} over $E$ is the Eilengberg-Moore category of $\mathcal{D}_E$.  A convex set may thus be thought of as a set $X$ together with a function mapping any finite tuple $\langle r_1, \ldots, r_n \rangle$ of elements of $M$ that sum to 1, and any tuple $\langle x_1, \ldots, x_n \rangle$ of elements of $X$, to an element
$r_1 x_1 + \cdots + r_n x_n$
of $X$.

\begin{theorem}
The distribution monad is a strong monad.  It is a commutative monad iff $E$ is commutative.
\end{theorem}

\begin{proof}
The tensorial strength $t_{AB} : A \times \mathcal{D}_E B \rightarrow \mathcal{D}_E (A \times B)$ is given by
\[ t_{AB}(a, \phi)(a',b) = \begin{cases}
\phi(b) & \mbox{if } a = a' \\
0 & \mbox{if } a \neq a'
\end{cases} \]
\end{proof}

\begin{corollary}
If $E$ is commutative, then $\Conv{E}$ is a symmetric monoidal category.
\end{corollary}

\begin{proof}
See \cite{Kock1972}.
\end{proof}

The convex set $A \otimes B$ consists of all sums $r_1 (a_1, b_1) + \cdots + r_n (a_n, b_n) \ (r_1 \ovee \cdots \ovee r_n = 1, a_i \in A, b_i \in B)$,
quotiented by the appropriate equivalence relation.  An affine function $f : A \otimes B \rightarrow C$ in $\Conv{M}$ is determined by the values $f(a,b)$ for $a \in A$ and $b \in B$

\begin{theorem}
The hom-functors $\Conv{E}[-,E] \dashv \EMod{E}[-,E] : \Conv{E} \rightleftarrows \EMod{E}^{\mathrm{op}}$ form an adjunction.
\end{theorem}

\begin{proof}
To appear in \cite{Jacobs}.  The special case $E = [0,1]$ was proved in \cite{Jacobs2010}.
\end{proof}

\section{Syntax and Rules of Deduction}

We begin with a system that represents a symmetric monoidal closed category with distributive coproducts, with an effect module of predicates over each object.

\newcommand{\lett}{\mathsf{let} \ }
\newcommand{\inn}{\ \mathsf{in} \ }
\newcommand{\of}{\ \mathsf{of} \ }
\newcommand{\case}{\mathsf{case} \ }
\newcommand{\measure}{\mathsf{measure} \ }
\newcommand{\eff}{\ \mathrm{eff}}
\[ \begin{array}{lrcl}
\mbox{Type} & A & ::= & A \otimes A \mid I \mid A + B \\
\mbox{Term} & M & ::= & x \mid M \otimes M \mid \lett x \otimes x = M \inn M \mid \langle \rangle \mid \\
& & & \inl{M} \mid \inr{M} \mid \\
& & & (\case M \of \inl{x} \mapsto M \mid \inr{x} \mapsto M) \mid \\
& & & (\measure \phi \mapsto M \mid \cdots \mid \phi \mapsto M) \\
\mbox{Effect} & \phi & ::= & 0 \mid \phi \ovee \phi \mid \phi^\bot \mid \phi \cdot \psi \mid (\case M \of \inl{x} \mapsto \phi \mid \inr{x} \mapsto \phi) \\
\mbox{Context} & \Gamma & ::= & \langle \rangle \mid \Gamma, x : A \\
\mbox{Judgement} & \mathcal{J} & ::= & \Gamma \vdash M : A \mid \Gamma \vdash M = N : A \mid \Gamma \vdash \phi \eff \mid \Gamma \vdash \phi \leq \psi
\end{array} \]

The intuition is as follows:
\begin{itemize}
\item
Each type represents a state space for a quantum computer at some stage of a calculation.  For example, the type $(\qbit \otimes \qbit) + (\qbit \otimes \qbit \otimes \qbit)$ represents a computer that has either two or three qubits in memory (depending on decisions earlier in the program).  (The type $\qbit$ will be introduced in Section \ref{section:qubits}.)

The type $I$ represents a singleton data type.  A term of type $A \otimes B$ is a pair consisting of a term of type $A$ and a term of type $B$ (possibly entangled).  A term of type $A + B$ is either a term of type $A$ or a term of type $B$ (with 'or' understood here classically).
\item
A term $M$ such that $\Gamma \vdash M : A$ represents a quantum algorithm that takes inputs as given by the context $\Gamma$, and returns an output of type $A$.

If the judgement $\Gamma \vdash M = N : A$ is derivable, then the algorithms $M$ and $N$ always produce the same output state given the same input state.
\item
An effect in context $\Gamma$ represents an observable measurement that may be performed on the system denoted by $\Gamma$.

The effect $0$ is the always false effect.  The effect $\phi \ovee \psi$ is the sum of $\phi$ and $\psi$, which may only be formed if $\phi$ and $\psi$ are orthogonal.  The effect $\phi^\bot$ is the orthocomplement of $\phi$.
\end{itemize}

We write
\begin{align*}
1 & \mbox{ for } 0^\bot \\
\ovee_{i=1}^n \phi_i & \mbox{ for } ((\cdots(\phi_1 \ovee \phi_2) \ovee \cdots) \ovee \phi_n \\
\measure_{i=1}^n \phi_i \mapsto M_i & \mbox{ for } \measure \phi_1 \mapsto M_1 \mid \cdots \mid \phi_n \mapsto M_n \\
\phi \perp \psi & \mbox{ for } \phi \leq \psi^\bot
\end{align*}
We write $\Gamma \vdash \phi \equiv \psi$ for the two judgements $\Gamma \vdash \phi \leq \psi$ and $\Gamma \vdash \psi \leq \phi$.

\newcommand{\Rovee}{($\ovee$)\xspace}
The rules of deduction are as follows.  

\paragraph{Note}
Note in particular the rule \Rovee.  For $\phi \ovee \psi$ to be a well-formed effect in context $\Gamma$, we must first have a derivation of $\Gamma \vdash \phi \perp \psi$, i.e.~$\Gamma \vdash \phi \leq \psi^\bot$.

\pagebreak

\paragraph{Structural Rule}

\newcommand{\Rexch}{(exchange)\xspace}
\begin{prooftree}
\AxiomC{$\Gamma, x : A, y : B, \Delta \vdash \mathcal{J}$}
\LeftLabel{(exch)}
\UnaryInfC{$\Gamma, y : B, x : A, \Delta \vdash \mathcal{J}$}
\end{prooftree}

\paragraph{Term Formation}

\newcommand{\Rvar}{(var)\xspace}
\begin{prooftree}
\AxiomC{}
\LeftLabel{\Rvar}
\RightLabel{($x : A \in \Gamma$)}
\UnaryInfC{$\Gamma \vdash x : A$}
\end{prooftree}

\newcommand{\Rotimes}{($\otimes$)\xspace}
\begin{prooftree}
\AxiomC{$\Gamma \vdash M : A$}
\AxiomC{$\Delta \vdash N : B$}
\LeftLabel{\Rotimes}
\BinaryInfC{$\Gamma, \Delta \vdash M \otimes N : A \otimes B$}
\end{prooftree}

\newcommand{\Rlet}{(let)\xspace}
\begin{prooftree}
\AxiomC{$\Gamma \vdash M : A \otimes B$}
\AxiomC{$\Delta, x : A, y : B \vdash N : C$}
\LeftLabel{\Rlet}
\BinaryInfC{$\Gamma, \Delta \vdash \lett x \otimes y = M \inn N : C$}
\end{prooftree}

\newcommand{\Runit}{($\langle \rangle$)\xspace}
\begin{prooftree}
\AxiomC{}
\LeftLabel{\Runit}
\UnaryInfC{$\Gamma \vdash \langle \rangle : I$}
\end{prooftree}

\newcommand{\Rinl}{(inl)\xspace}
\newcommand{\Rinr}{(inr)\xspace}
\begin{center}
\AxiomC{$\Gamma \vdash M : A$}
\LeftLabel{\Rinl}
\UnaryInfC{$\Gamma \vdash \inl{M} : A + B$}
\DisplayProof
\qquad
\AxiomC{$\Gamma \vdash M : B$}
\LeftLabel{\Rinr}
\UnaryInfC{$\Gamma \vdash \inr{M} : A + B$}
\DisplayProof
\end{center}

\newcommand{\Rcase}{(case)\xspace}
\begin{prooftree}
\AxiomC{$\Gamma \vdash M : A + B$}
\AxiomC{$\Delta, x : A \vdash N : C$}
\AxiomC{$\Delta, y : B \vdash P : C$}
\LeftLabel{\Rcase}
\TrinaryInfC{$\Gamma, \Delta \vdash \case M \of \inl{x} \mapsto N \mid \inr{y} \mapsto P : C$}
\end{prooftree}

\newcommand{\Rmeasure}{(measure)\xspace}
\begin{prooftree}
\AxiomC{$\Gamma \vdash 1 \leq \ovee_{i=1}^n \phi_i$}
\AxiomC{$\Delta \vdash M_i : A \qquad (1 \leq i \leq n)$}
\LeftLabel{\Rmeasure}
\BinaryInfC{$\Gamma, \Delta \vdash \measure_{i=1}^n \phi_i \mapsto M_i: A$}
\end{prooftree}

\paragraph{Equality of Terms}

\newcommand{\Rref}{(ref)\xspace}
\newcommand{\Rsym}{(sym)\xspace}
\begin{center}
\AxiomC{$\Gamma \vdash M : A$}
\LeftLabel{\Rref}
\UnaryInfC{$\Gamma \vdash M = M : A$}
\DisplayProof
\qquad
\AxiomC{$\Gamma \vdash M = N : A$}
\LeftLabel{\Rsym}
\UnaryInfC{$\Gamma \vdash N = M : A$}
\DisplayProof
\end{center}

\newcommand{\Rtrans}{(trans)\xspace}
\begin{prooftree}
\AxiomC{$\Gamma \vdash M = N : A$}
\AxiomC{$\Gamma \vdash N = P : A$}
\LeftLabel{(trans)}
\BinaryInfC{$\Gamma \vdash M = P : A$}
\end{prooftree}

\subparagraph{Congruences}

\newcommand{\Rtimeseq}{($\otimes$-eq)\xspace}
\begin{prooftree}
\AxiomC{$\Gamma \vdash M = M' : A$}
\AxiomC{$\Delta \vdash N = N' : B$}
\LeftLabel{\Rtimeseq}
\BinaryInfC{$\Gamma, \Delta \vdash M \otimes N = M' \otimes N' : A
\otimes B$}
\end{prooftree}

\newcommand{\Rleteq}{(let-eq)\xspace}
\begin{prooftree}
\AxiomC{$\Gamma \vdash M = M' : A \otimes B$}
\AxiomC{$\Delta, x : A, y : B \vdash N = N' : C$}
\LeftLabel{\Rleteq}
\BinaryInfC{$\Gamma, \Delta \vdash (\lett x \otimes y = M \inn N) = (\lett x \otimes y = M' \inn N') : C$}
\end{prooftree}

\newcommand{\Rinleq}{(inl-eq)\xspace}
\newcommand{\Rinreq}{(inr-eq)\xspace}
\begin{center}
\AxiomC{$\Gamma \vdash M = N : A$}
\LeftLabel{\Rinleq}
\UnaryInfC{$\Gamma \vdash \inl{M} = \inl{N} : A + B$}
\DisplayProof
\qquad
\AxiomC{$\Gamma \vdash M = N : B$}
\LeftLabel{\Rinreq}
\UnaryInfC{$\Gamma \vdash \inr{M} = \inr{N} : A + B$}
\DisplayProof
\end{center}

\newcommand{\Rcaseeq}{(case-eq)\xspace}
\begin{prooftree}
\AxiomC{$\Gamma \vdash M = M' : A + B$}
\AxiomC{$\Delta, x : A \vdash N = N' : C$}
\AxiomC{$\Delta, y : B \vdash P = P' : C$}
\LeftLabel{\Rcaseeq}
\TrinaryInfC{$
\begin{array}{c}
\Gamma, \Delta \vdash (\case M \of \inl{x} \mapsto N \mid \inr{y} \mapsto P) \\
= (\case M' \of \inl{x} \mapsto N' \mid \inr{y} \mapsto P') : C
\end{array}
$}
\end{prooftree}

\newcommand{\Rmeaseq}{(measure-eq)\xspace}
\begin{center}
\AxiomC{$\Gamma \vdash 1 \leq \ovee_{i=1}^n \phi_i$}
\AxiomC{$\Gamma \vdash \phi_i \equiv \psi_i \quad (1 \leq i \leq n)$}
\noLine
\BinaryInfC{$\Delta \vdash M_i = N_i : A \quad (1 \leq i \leq n)$}
\LeftLabel{\Rmeaseq}
\UnaryInfC{$\Gamma, \Delta \vdash (\measure_{i=1}^n \phi_i \mapsto M_i) = (\measure_{i=1}^n \psi_i \mapsto N_i) : A$}
\DisplayProof
\end{center}

\subparagraph{$\beta$-conversions}

\newcommand{\Rbetaotimes}{($\beta \otimes$)\xspace}
\begin{prooftree}
\AxiomC{$\Gamma \vdash M : A$}
\AxiomC{$\Delta \vdash N : B$}
\AxiomC{$\Theta, x : A, y : B \vdash P : C$}
\LeftLabel{\Rbetaotimes}
\TrinaryInfC{$\Gamma, \Delta, \Theta \vdash (\lett x \otimes y = M \otimes N \inn P) =
[M / x, N / y]P : C$}
\end{prooftree}

\newcommand{\Rbetaplusone}{($\beta +_1$)\xspace}
\begin{prooftree}
\AxiomC{$\Gamma \vdash M : A$}
\AxiomC{$\Delta, x : A \vdash N : C$}
\AxiomC{$\Delta, y : B \vdash P : C$}
\LeftLabel{\Rbetaplusone}
\TrinaryInfC{$\Gamma, \Delta \vdash \case \inl{M} \of \inl{x} \mapsto N \mid \inr{y} \mapsto P = [M/x]N : C$}
\end{prooftree}

\newcommand{\Rbetaplustwo}{($\beta +_2$)\xspace}
\begin{prooftree}
\AxiomC{$\Gamma \vdash M : B$}
\AxiomC{$\Delta, x : A \vdash N : C$}
\AxiomC{$\Delta, y : B \vdash P : C$}+
\LeftLabel{\Rbetaplustwo}
\TrinaryInfC{$\Gamma, \Delta \vdash \case \inr{M} \of \inl{x} \mapsto N \mid \inr{y} \mapsto P = [M/y]P : C$}
\end{prooftree}

\subparagraph{$\eta$-conversions}

\newcommand{\Retaotimes}{($\eta \otimes$)\xspace}
\newcommand{\RetaI}{($\eta I$)\xspace}
\begin{center}
\AxiomC{$\Gamma \vdash M : A \otimes B$}
\LeftLabel{\Retaotimes}
\UnaryInfC{$\Gamma \vdash M = \lett x \otimes y = M \inn x \otimes y : A \otimes B$}
\DisplayProof
\qquad
\AxiomC{$\Gamma \vdash M : I$}
\LeftLabel{\RetaI}
\UnaryInfC{$\Gamma \vdash M = \langle\rangle  : I$}
\DisplayProof
\end{center}

\newcommand{\Retaplus}{($\eta +$)\xspace}
\begin{prooftree}
\AxiomC{$\Gamma \vdash M : A + B$}
\LeftLabel{\Retaplus}
\UnaryInfC{$\Gamma \vdash M = \case M \of \inl{x} \mapsto \inl{x} \mid \inr{y} \mapsto \inr{y} : A + B$}
\end{prooftree}

\subparagraph{Commuting Conversions}

\newcommand{\Rletlet}{(let-commute)\xspace}
\begin{prooftree}
\AxiomC{$\Gamma \vdash M : A \otimes B$}
\AxiomC{$\Delta, x : A, y : B \vdash N : C \otimes D$}
\noLine
\BinaryInfC{$\Theta, t : C, u : D \vdash P : E$}
\LeftLabel{\Rletlet}
\UnaryInfC{$
\begin{array}{c}
\Gamma, \Delta, \Theta \vdash (\lett x \otimes y = M \inn \lett t \otimes u = N \inn P) \\
= (\lett t \otimes u = \lett x \otimes y = M \inn N \inn P) : E
\end{array}
$}
\end{prooftree}

\newcommand{\Rletcase}{(let-case)\xspace}
\begin{prooftree}
\AxiomC{$\Gamma \vdash M : A + B$}
\AxiomC{$\Delta, x : A \vdash N : C \otimes D$}
\noLine
\BinaryInfC{$\Delta, y : B \vdash P : C \otimes D \qquad
\Theta, z : C, t : D \vdash Q : E$}
\LeftLabel{\Rletcase}
\UnaryInfC{$
\begin{array}{c}
\Gamma, \Delta, \Theta \vdash (\lett z \otimes t = \case M \of \inl{x} \mapsto N \mid \inr{y} \mapsto P \inn Q) \\
= \mathsf{case}\ M\ \mathsf{of}\ \inl{x} \mapsto \lett z \otimes t = N \inn Q \mid \\
\inr{y} \mapsto \lett z \otimes t = P \inn Q : E
\end{array}
$}
\end{prooftree}

\newcommand{\Rletotimes}{(let-$\otimes$)\xspace}
\begin{prooftree}
\AxiomC{$\Gamma \vdash M : A \otimes B$}
\AxiomC{$\Delta, x : A, y : B \vdash N : C$}
\AxiomC{$\Theta \vdash P : D$}
\LeftLabel{\Rletotimes}
\TrinaryInfC{$\Gamma, \Delta, \Theta \vdash (\lett x \otimes y = M \inn N) \otimes P = \lett x \otimes y = M \inn (N \otimes P)$}
\end{prooftree}

\newcommand{\Rcasecase}{(case-commute)\xspace}
\begin{prooftree}
\AxiomC{$\Gamma \vdash M : A + B$}
\AxiomC{$\Delta, x : A \vdash N : C + D$}
\AxiomC{$\Delta, y : B \vdash P : C + D$}
\noLine
\TrinaryInfC{$\Theta, z : C \vdash Q : E \qquad \qquad \Theta, t : D \vdash R : E$}
\LeftLabel{\Rcasecase}
\UnaryInfC{$
\begin{array}{c}
\Gamma, \Delta, \Theta \vdash 
\case M \of \inl{x} \mapsto (\case N \of \inl{z} \mapsto Q \mid \inr{t} \mapsto R) \mid \\
\inr{y} \mapsto (\case P \of \inl{z} \mapsto Q \mid \inr{t} \mapsto R) \mid \\
= \case (\case M \of \inl{x} \mapsto N \mid \inr{y} \mapsto P) \of \inl{z} \mapsto Q \mid \inr{t} \mapsto R : E
\end{array}
$}
\end{prooftree}

\newcommand{\Rcaseotimes}{(case-$\otimes$)\xspace}
\begin{prooftree}
\AxiomC{$\Gamma \vdash Q : A + B$}
\AxiomC{$\Delta, a : A \vdash M : C$}
\AxiomC{$\Delta, b : B \vdash N : C$}
\AxiomC{$\Theta \vdash P : D$}
\LeftLabel{\Rcaseotimes}
\QuaternaryInfC{$
\begin{array}{c}
\Gamma, \Delta, \Theta \vdash (\case Q \of \inl{a} \mapsto M \mid \inr{b} \mapsto N) \otimes P \\
= \case Q \of \inl{a} \mapsto M \otimes P \mid \inr{b} \mapsto N \otimes P : C \otimes D
\end{array}
$}
\end{prooftree}

\subparagraph{Rules for Measurement}

\newcommand{\Rmeasureperm}{(measure-perm)\xspace}
\begin{prooftree}
\AxiomC{$\Gamma \vdash 1 \leq \ovee_{i=1}^n \phi_i$}
\AxiomC{$\Delta \vdash M_x : A \quad (1 \leq i \leq n)$}
\LeftLabel{\Rmeasureperm}
\RightLabel{($p$ a permutation of $\{1, \ldots, n\}$)}
\BinaryInfC{$
\begin{array}{c}
\Gamma, \Delta \vdash (\measure_{i=1}^n \phi_i \mapsto M_i) \\
=(\measure_{i=1}^n \phi_{p(i)} \mapsto M_{p(i)}) : A
\end{array}
$}
\end{prooftree}

\newcommand{\Rmeasurezero}{(measure-0)\xspace}
\begin{prooftree}
\AxiomC{$\Gamma \vdash 1 \leq \ovee_{i=1}^n \phi_i$}
\noLine
\UnaryInfC{$\Delta \vdash M_i : A \qquad (1 \leq i \leq n+1)$}
\LeftLabel{\Rmeasurezero}
\UnaryInfC{$
\begin{array}{c}
\Gamma, \Delta \vdash (\measure \phi_1 \mapsto M_1 \mid \cdots \mid \phi_n \mapsto M_n \mid 0 \mapsto M_{n+1}) \\
= \measure \phi_1 \mapsto M_1 \mid \cdots \mid \phi_n \mapsto M_n : A
\end{array}
$}
\end{prooftree}

\newcommand{\Rmeasureone}{(measure-1)\xspace}
\begin{prooftree}
\AxiomC{$\Gamma \vdash M : A$}
\LeftLabel{\Rmeasureone}
\UnaryInfC{$\Gamma \vdash (\mathsf{measure}\  1 \mapsto M) = M : A$}
\end{prooftree}

\newcommand{\Rmeasureplus}{(measure-plus)\xspace}
\begin{prooftree}
\AxiomC{$\vdash 1 \leq \phi \ovee \psi \ovee \chi_1 \ovee \cdots \ovee \chi_n$}
\AxiomC{$\Gamma \vdash M : A$}
\AxiomC{$\Gamma \vdash P_1 : A \qquad \cdots \qquad \Gamma \vdash P_n : A$}
\LeftLabel{\Rmeasureplus}
\TrinaryInfC{$
\begin{array}{c}
\Gamma \vdash (\mathsf{measure}\ \phi \ovee \psi \mapsto M \mid \chi_1 \mapsto P_1 \mid \cdots \mid \chi_n \mapsto P_n) \\
= (\mathsf{measure}\ \phi \mapsto M \mid \psi \mapsto M \mid \chi_1 \mapsto P_1 \mid \cdots \mid \chi_n \mapsto P_n)
\end{array}
$}
\end{prooftree}

\newcommand{\Rmeasurecase}{(measure-case)\xspace}
\begin{prooftree}
\AxiomC{$\Gamma, x : A \vdash 1 \leq \ovee_{i=1}^n \phi_i$}
\AxiomC{$\Gamma, y : B \vdash 1 \leq \ovee_{i=1}^n \psi_i$}
\noLine
\BinaryInfC{$\Delta \vdash M : A + B \qquad \Theta \vdash N_i : C \quad (1 \leq i \leq n)$}
\LeftLabel{\Rmeasurecase}
\UnaryInfC{$
\begin{array}{c}
\Gamma, \Delta, \Theta \vdash
\measure_{i=1}^n (\case M \of \inl{x} \mapsto \phi_i \mid \inr{y} \mapsto \psi_i) \mapsto N_i \\
= \case M \of \inl{x} \mapsto (\measure_{i=1}^n \phi_i \mapsto N_i) \mid \\
\inr{y} \mapsto (\measure_{i=1}^n \psi_i \mapsto N_i)
\end{array}
$}
\end{prooftree}

\paragraph{Effect Formation}

\newcommand{\Rzero}{(0)\xspace}
\newcommand{\Rbot}{($\bot$)\xspace}
\begin{center}
\AxiomC{}
\LeftLabel{\Rzero}
\UnaryInfC{$\Gamma \vdash 0 \eff$}
\DisplayProof
\qquad
\AxiomC{$\Gamma \vdash \phi \eff$}
\LeftLabel{\Rbot}
\UnaryInfC{$\Gamma \vdash \phi^\bot \eff$}
\DisplayProof
\qquad
\AxiomC{$\Gamma \vdash \phi \perp \psi$}
\LeftLabel{\Rovee}
\UnaryInfC{$\Gamma \vdash \phi \ovee \psi \eff$}
\DisplayProof
\end{center}

\begin{prooftree}
\AxiomC{$\vdash \phi \eff$}
\AxiomC{$\Gamma \vdash \psi \eff$}
\LeftLabel{(mult)}
\BinaryInfC{$\Gamma \vdash \phi \cdot \psi \eff$}
\end{prooftree}

\begin{prooftree}
\AxiomC{$\Gamma, x : A \vdash \phi \eff$}
\AxiomC{$\Gamma, y : B \vdash \psi \eff$}
\AxiomC{$\Delta \vdash M : A + B$}
\LeftLabel{(case)}
\TrinaryInfC{$\Gamma, \Delta \vdash \case M \of \inl{x} \mapsto \phi \mid \inr{y} \mapsto \psi \eff$}
\end{prooftree}

\paragraph{Derivability}

\newcommand{\Rleqref}{($\leq$-ref)\xspace}
\newcommand{\Rleqtrans}{($\leq$-trans)\xspace}
\begin{center}
\AxiomC{$\Gamma \vdash \phi \eff$}
\LeftLabel{\Rleqref}
\UnaryInfC{$\Gamma \vdash \phi \leq \phi$}
\DisplayProof
\qquad
\AxiomC{$\Gamma \vdash \phi \leq \psi$}
\AxiomC{$\Gamma \vdash \psi \leq \chi$}
\LeftLabel{\Rleqtrans}
\BinaryInfC{$\Gamma \vdash \phi \leq \chi$}
\DisplayProof
\end{center}

\newcommand{\Rzeroleq}{(0-$\leq$)\xspace}
\begin{prooftree}
\AxiomC{$\Gamma \vdash \phi \eff$}
\LeftLabel{\Rzeroleq}
\UnaryInfC{$\Gamma \vdash 0 \leq \phi$}
\end{prooftree}

\newcommand{\Rbotanti}{($\bot$-antitone)\xspace}
\newcommand{\Rprebotbot}{($\bot \bot_0$)\xspace}
\begin{center}
\AxiomC{$\Gamma \vdash \phi \leq \psi$}
\LeftLabel{\Rbotanti}
\UnaryInfC{$\Gamma \vdash \psi^\bot \leq \phi^\bot$}
\DisplayProof
\qquad
\AxiomC{$\Gamma \vdash \phi \eff$}
\LeftLabel{\Rprebotbot}
\UnaryInfC{$\Gamma \vdash \phi \leq \phi^{\bot \bot}$}
\DisplayProof
\end{center}

\newcommand{\Rleqovee}{($\leq$-$\ovee$)\xspace}
\newcommand{\Roveemono}{($\ovee$-mono)\xspace}
\begin{center}
\AxiomC{$\Gamma \vdash \phi \perp \psi$}
\LeftLabel{\Rleqovee}
\UnaryInfC{$\Gamma \vdash \phi \leq \phi \ovee \psi$}
\DisplayProof
\qquad
\AxiomC{$\Gamma \vdash \phi \leq \psi$}
\AxiomC{$\Gamma \vdash \psi \leq \chi^\bot$}
\LeftLabel{\Roveemono}
\BinaryInfC{$\Gamma \vdash \phi \ovee \chi \leq \psi \ovee \chi$}
\DisplayProof
\end{center}

\newcommand{\Roveeprecomm}{($\ovee$-comm$_\leq$)\xspace}
\newcommand{\Rperprot}{($\perp$-rotate)\xspace}
\begin{center}
\AxiomC{$\Gamma \vdash \phi \perp \psi$}
\LeftLabel{\Roveeprecomm}
\UnaryInfC{$\Gamma \vdash \phi \ovee \psi \leq \psi \ovee \phi$}
\DisplayProof
\qquad
\AxiomC{$\Gamma \vdash \phi \ovee \psi \perp \chi$}
\LeftLabel{\Rperprot}
\UnaryInfC{$\Gamma \vdash \psi \ovee \chi \perp \phi$}
\DisplayProof
\end{center}

\newcommand{\Roveepreassoc}{($\ovee$-assoc$_\leq$)\xspace}
\newcommand{\Roveezero}{($\ovee$-0)\xspace}
\begin{center}
\AxiomC{$\Gamma \vdash \phi \ovee \psi \perp \chi$}
\LeftLabel{\Roveepreassoc}
\UnaryInfC{$\Gamma \vdash \phi \ovee (\psi \ovee \chi) \leq (\phi \ovee \psi) \ovee \chi$}
\DisplayProof
\qquad
\AxiomC{$\Gamma \vdash \phi \eff$}
\LeftLabel{\Roveezero}
\UnaryInfC{$\Gamma \vdash \phi \ovee 0 \leq \phi$}
\DisplayProof
\end{center}

\newcommand{\Rortho}{(ortho$_1$)}
\newcommand{\Rorthotwo}{(ortho$_2$)}
\begin{center}
\AxiomC{$\Gamma \vdash 1 \leq \phi \ovee \psi$}
\LeftLabel{\Rortho}
\UnaryInfC{$\Gamma \vdash \psi^\bot \leq \phi$}
\DisplayProof
\qquad
\AxiomC{$\Gamma \vdash \phi \eff$}
\LeftLabel{\Rorthotwo}
\UnaryInfC{$\Gamma \vdash 1 \leq \phi \ovee \phi^\bot$}
\DisplayProof
\end{center}

\begin{center}
\AxiomC{$\vdash \phi \perp \psi$}
\AxiomC{$\Gamma \vdash \chi \eff$}
\LeftLabel{(dist$_L$)}
\BinaryInfC{$
\begin{array}{c}
\Gamma \vdash \phi \cdot \chi \perp \psi \cdot \chi \\
\Gamma \vdash (\phi \ovee \psi) \cdot \chi \equiv \phi \cdot \chi \ovee \psi \cdot \chi
\end{array}
$}
\DisplayProof
\qquad
\AxiomC{$\vdash \phi \eff$}
\AxiomC{$\Gamma \vdash \psi \perp \chi$}
\LeftLabel{(dist$_R$)}
\BinaryInfC{$
\begin{array}{c}
\Gamma \vdash \phi \cdot \psi \perp \phi \cdot \chi\\
\Gamma \vdash \phi \cdot (\psi \ovee \chi) \equiv \phi \cdot \psi \ovee \phi \cdot \chi
\end{array}
$}
\DisplayProof
\end{center}

\begin{center}
\AxiomC{$\Gamma \vdash \phi \eff$}
\LeftLabel{(unit$_L$)}
\UnaryInfC{$\Gamma \vdash 1 \cdot \phi \equiv \phi$}
\DisplayProof
\qquad
\AxiomC{$\vdash \phi \eff$}
\LeftLabel{(unit$_R$)}
\UnaryInfC{$\Gamma \vdash \phi \cdot 1 \equiv \phi$}
\DisplayProof
\qquad
\AxiomC{$\vdash \phi \eff$}
\AxiomC{$\vdash \psi \eff$}
\AxiomC{$\Gamma \vdash \chi \eff$}
\LeftLabel{(assoc)}
\TrinaryInfC{$\Gamma \vdash \phi \cdot (\psi \cdot \chi) \equiv (\phi \cdot \psi) \cdot \chi$}
\DisplayProof
\end{center}

\begin{prooftree}
\AxiomC{$\vdash \phi \eff$}
\AxiomC{$\vdash \psi \eff$}
\LeftLabel{(comm)}
\BinaryInfC{$\vdash \phi \cdot \psi \equiv \psi \cdot \phi$}
\end{prooftree}

\newcommand{\Rcasecong}{(case-cong)\xspace}
\begin{center}
\AxiomC{$\Gamma, x : A \vdash \phi \eff$}
\AxiomC{$\Gamma, y : B \vdash \psi \eff$}
\AxiomC{$\Delta \vdash M = N : A + B$}
\LeftLabel{\Rcasecong}
\TrinaryInfC{$\Gamma, \Delta \vdash \case M \of \inl{x} \mapsto \phi \mid \inr{y} \mapsto \psi \equiv \case N \of \inl{x} \mapsto \phi \mid \inr{y} \mapsto \psi$}
\DisplayProof
\end{center}

\newcommand{\Rcasemono}{(case-mono)}
\begin{center}
\AxiomC{$\Gamma, x : A \vdash \phi \leq \phi'$}
\AxiomC{$\Gamma, y : B \vdash \psi \leq \psi'$}
\AxiomC{$\Delta \vdash M : A + B$}
\LeftLabel{\Rcasemono}
\TrinaryInfC{$\Gamma, \Delta \vdash \case M \of \inl{x} \mapsto \phi \mid \inr{y} \mapsto \psi \leq \case M \of \inr{x} \mapsto \phi' \mid \inr{y} \mapsto \psi'$}
\DisplayProof
\end{center}

\newcommand{\Rbetaoneeff}{($\beta +_1$-eff)\xspace}
\begin{center}
\AxiomC{$\Gamma, x : A \vdash \phi \eff$}
\AxiomC{$\Gamma, y : B \vdash \psi \eff$}
\AxiomC{$\Delta \vdash M : A$}
\LeftLabel{\Rbetaoneeff}
\TrinaryInfC{$\Gamma, \Delta \vdash \case \inl{M} \of \inl{x} \mapsto \phi \mid \inr{y} \mapsto \psi \equiv [M/x]\phi$}
\DisplayProof
\end{center}

\newcommand{\Rbetatwoeff}{($\beta +_2$-eff)\xspace}
\begin{center}
\AxiomC{$\Gamma, x : A \vdash \phi \eff$}
\AxiomC{$\Gamma, y : B \vdash \psi \eff$}
\AxiomC{$\Delta \vdash M : B$}
\LeftLabel{\Rbetatwoeff}
\TrinaryInfC{$\Gamma, \Delta \vdash \case \inr{M} \of \inl{x} \mapsto \phi \mid \inr{y} \mapsto \psi \equiv [M/x]\psi$}
\DisplayProof
\end{center}

\newcommand{\Rcaseeta}{($\eta+$-eff)\xspace}
\begin{center}
\AxiomC{$\Gamma, z : A + B \vdash \phi \eff$}
\LeftLabel{\Rcaseeta}
\UnaryInfC{$\Gamma, z : A + B \vdash \phi \equiv \case z \of \inl{x} \mapsto [\inl{x}/z]\phi \mid \inr{y} \mapsto [\inr{y} / z] \phi$}
\DisplayProof
\end{center}

\newcommand{\Rcaseovee}{(case-$\ovee$)\xspace}
\begin{center}
\AxiomC{$\Gamma, x : A \vdash \phi \perp \phi'$}
\AxiomC{$\Gamma, y : B \vdash \psi \perp \psi'$}
\AxiomC{$\Delta \vdash M : A + B$}
\LeftLabel{\Rcaseovee}
\TrinaryInfC{$
\begin{array}{c}
\Gamma, \Delta \vdash \case M \of \inl{x} \mapsto (\phi \ovee \phi') \mid \inr{y} \mapsto (\psi \ovee \psi') \\
\equiv
(\case M \of \inl{x} \mapsto \phi \mid \inr{y} \mapsto \psi) \ovee (\case M \of \inl{x} \mapsto \phi' \mid \inr{y} \mapsto \psi')
\end{array}
$}
\DisplayProof
\end{center}

\newcommand{\Rcasebot}{(case-$\bot$)\xspace}
\begin{center}
\AxiomC{$\Gamma, x : A \vdash \phi \eff$}
\AxiomC{$\Gamma, y : B \vdash \psi \eff$}
\AxiomC{$\Delta \vdash M : A + B$}
\LeftLabel{\Rcasebot}
\TrinaryInfC{$\Gamma, \Delta \vdash \case M \of \inl{x} \mapsto \phi^\bot \mid \inr{y} \mapsto \psi^\bot \equiv
(\case M \of \inl{x} \mapsto \phi \mid \inr{y} \mapsto \psi)^\bot$}
\DisplayProof
\end{center}

\begin{center}
\AxiomC{$\Gamma \vdash M : A + B$}
\AxiomC{$\Delta, x : A \vdash \phi \leq \chi$}
\AxiomC{$\Delta, y : B \vdash \psi \leq \chi$}
\LeftLabel{(case-$\leq$)}
\TrinaryInfC{$\Gamma, \Delta \vdash (\case M \of \inl{x} \mapsto \phi \mid \inr{y} \mapsto \psi) \leq \chi$}
\DisplayProof
\end{center}

\newcommand{\Rcasemult}{(case-times)\xspace}
\begin{center}
\AxiomC{$\Gamma, x : A \vdash \phi \eff$}
\AxiomC{$\Gamma, y : B \vdash \psi \eff$}
\AxiomC{$\Delta \vdash M : A + B$}
\AxiomC{$\vdash \chi \eff$}
\LeftLabel{\Rcasemult}
\QuaternaryInfC{$
\begin{array}{c}
\Gamma, \Delta \vdash (\case M \of \inl{x} \mapsto \chi \cdot \phi \mid \inr{y} \mapsto \chi \cdot \psi) \\
\equiv
\chi \cdot \case M \of \inl{x} \mapsto \phi \mid \inr{y} \mapsto \psi
\end{array}
$}
\DisplayProof
\end{center}

\subsection{Metatheorems}
\label{section:metatheorems}

We can prove the following properties, which show that the typing system is well behaved.

\begin{lemma}$ $
\begin{enumerate}
\item
\textbf{Substitution}
If $\Gamma \vdash M : A$ and $\Delta, x : A, \Delta' \vdash \mathcal{J}$ then $\Delta, \Gamma, \Delta' \vdash [M/x]\mathcal{J}$.
\item
\textbf{Weakening}
If $\Gamma \vdash M : A$ and $\Gamma \subseteq \Delta$ then $\Delta \vdash M : A$.
\end{enumerate}
\end{lemma}

\begin{proof}
The proof is straightforward, by induction on derivations.
\end{proof}

\begin{lemma}[Equation Validity]
\label{lm:eqval}
$ $
\begin{enumerate}
\item
If $\Gamma \vdash M = N : A$ then $\Gamma \vdash M : A$ and $\Gamma \vdash N : A$.
\item
If $\Gamma \vdash \phi \leq \psi$ then $\Gamma \vdash \phi \prop$ and $\Gamma \vdash \psi \prop$.
\end{enumerate}
\end{lemma}

\begin{proof}
Let QPEL$'$ be the system where the rule \Rovee is replaced with
\newcommand{\Roveep}{($\ovee'$)\xspace}
\begin{prooftree}
\AxiomC{$\Gamma \vdash \phi \perp \psi$}
\AxiomC{$\Gamma \vdash \phi \eff$}
\AxiomC{$\Gamma \vdash \psi \eff$}
\LeftLabel{\Roveep}
\TrinaryInfC{$\Gamma \vdash \phi \ovee \psi \eff$}
\end{prooftree}
It is straightforward to prove that QPEL$'$ satisfies Equation Validity.
It follows that the derivable judgements of QPEL and QPEL$'$ are the same, and hence that QPEL satisfies Equation Validity.
\end{proof}

%

\begin{lemma}[Functionality]
$ $
\begin{enumerate}
\item
If $\Gamma \vdash M = N : A$ and $\Delta, x : A \vdash P : B$ then
$\Gamma, \Delta \vdash [M/x]P = [N/x]P : B$.
\item
If $\Gamma \vdash M = N : A$ and $\Delta, x : A \vdash \phi \prop$ then $\Gamma, \Delta \vdash [M/x]\phi \equiv [N/x]\phi$.
\end{enumerate}
\end{lemma}

\begin{proof}
Let QPEL$''$ be the system where \Rmeasure is replaced with the rule
\begin{prooftree}
\AxiomC{$\Gamma \vdash 1 \leq \ovee_{i=1}^n \phi_i$}
\AxiomC{$\Gamma \vdash \ovee_{i=1}^n \phi_i \eff$}
\AxiomC{$\Delta \vdash M_i : A \qquad (1 \leq i \leq n)$}
\LeftLabel{(measure$''$)}
\TrinaryInfC{$\Gamma, \Delta \vdash \measure_{i=1}^n \phi_i \mapsto M_i : A$}
\end{prooftree}

We can prove that QPEL$''$ satisfies Equation Validity, using the same proof technique as Lemma \ref{lm:eqval}.  It follows that QPEL and QPEL$''$ have the same derivable judgements.  It is straightforward to prove that QPEL$''$ satisfies Functionality, and so it follows that QPEL satisfies Functionality.
\end{proof}

\begin{lemma}
\begin{enumerate}
\item
\label{ltt1}
If $\Gamma, x : A, y : B \vdash M : C$, $\Delta \vdash N : A \otimes B$ and $\Theta,
z : C \vdash P : D$, then
\[ \Gamma, \Delta, \Theta \vdash [\lett x \otimes y = N \inn M / z] P =
(\lett x \otimes y = N \inn [M/z]P) : D \]


\item
\label{ltt2}
If $\Gamma \vdash M : A + B$, $\Delta, x : A \vdash N : C$, $\Delta, y : B \vdash P : C$ and $\Theta, z : C \vdash Q : D$, then
\begin{align*}
\Gamma, \Delta, \Theta \vdash & [ \case M \of \inl{x} \mapsto N \mid \inr{y} \mapsto P / z] Q \\
& = \case M \of \inl{x} \mapsto [N/z]Q \mid \inr{y} \mapsto [P/z]Q : D
\end{align*}

\end{enumerate}
\end{lemma}

\begin{proof}
The proof of this lemma involves noting that local definitions can be defined from the rules for $I$ and $\otimes$, which to the best of my knowledge is a new result about linear type theory.

If $\Gamma \vdash M : A$ and $\Delta, x : A \vdash N : B$, we define the term $\lett x = M \inn N$ to be
\[ \lett x \otimes y = M \otimes \langle \rangle \inn N \]
so $\Gamma, \Delta \vdash \lett x = M \inn N : B$ and
\[ \Gamma, \Delta \vdash (\lett x = M \inn N) = [M/x]N : B \enspace . \]

From the rules of derivation in QPEL, we can show that:
\begin{itemize}
\item
If $\Gamma \vdash N : A \otimes B$ and $\Delta, x : A, y : B \vdash M : C$ and $\Theta, z : C \vdash P : D$ then
\[ \Gamma, \Delta, \Theta \vdash (\lett z = \lett x \otimes y = N \inn M \inn P) = (\lett x \otimes y = N \inn \lett z = M \inn P) : D \]
\item
If $\Gamma \vdash M : A + B$, $\Delta, x : A \vdash N : C$, $\Delta, y : B \vdash P : C$, and $\Theta, z : C \vdash Q : D$, then
\begin{align*}
 \Gamma, \Delta, \Theta \vdash & \lett z = (\case M \of \inl{x} \mapsto N \mid \inr{y} \mapsto P) \inn Q \\ 
 & = \case M \of \inl{x} \mapsto \lett z = N \inn Q \mid \inr{y} \mapsto \lett z = P \inn Q : D
 \end{align*}
 \end{itemize}
The result then follows.
\end{proof}

%

\section{Semantics}

\subsection{State and Effect Triangles}

Let $E$ be a commutative effect monoid.
Recall the adjunction $\Conv{E}[-,E] \dashv \EMod{E}[-,E] : \Conv{E} \rightleftarrows \EMod{E}^{\mathrm{op}}$.

\newcommand{\meas}{\operatorname{meas}}
\begin{df}[State-and-Effect Triangle]
A \emph{state-and-effect triangle} consists of:
\begin{itemize}
\item
a symmetric monoidal category $\mathcal{V}$ with binary coproducts that distribute over the tensor, such that the tensor unit is terminal;
\item
an effect monoid $E$;
\item
a functor $P : \mathcal{V} \rightarrow \EMod{E}^\op$ that preserves finite coproducts and the terminal object;
\item
a symmetric monoidal functor $S : \mathcal{V} \rightarrow \Conv{E}$;
\item
given a finite set $r_1, \ldots, r_n \in PA$ such that $r_1 \ovee \cdots \ovee r_n = 1$, an arrow $\meas_A(r_1, \ldots, r_n) : A \rightarrow n \cdot I$ in $\mathcal{V}$;
\item
a natural transformation $\alpha : P \rightarrow \Conv{E}[S - , E]$;
\item
a natural transformation $\beta : S \rightarrow \EMod{E}[P - , E]$;
\end{itemize}
\pagebreak
such that
\begin{enumerate}
\item
given a permutation $p$ on $\{1, \ldots, n\}$, we have
\[ \meas_A(r_{p(1)}, \ldots, r_{p(n)}) = \pi_p \circ \meas_A(p_1, \ldots, p_n) \]
where $\pi_p : n \cdot I \rightarrow n \cdot I$ satisfies
\[ \pi_p \circ \kappa_i = \kappa_{p(i)} \]
\item
$\meas_A(p_1, \ldots, p_n, 0) = \kappa_1 \circ \meas_A(p_1, \ldots, p_n) : A \rightarrow n \cdot I \rightarrow (n + 1) \cdot I$
\item
$\meas_A(p \ovee q, r_1, \ldots, r_n) = [\kappa_1, \kappa_1, \kappa_2, \ldots, \kappa_{n+1}] \circ \meas_A(p, q, r_1, \ldots, r_n)$
\item
$\meas_A$ is natural in $A$; i.e. given $f : A \rightarrow B$,
\[ \meas_A(r_1, \ldots, r_n) \circ f = \meas_B(P f (r_1), \ldots, P f(r_n)) \]
\item
$\alpha_A(p)(x) = \beta_A(x)(p)$ for all $A$, $x$, $p$.
\end{enumerate}
\end{df}

We think of the arrows in $\mathcal{V}$ as \emph{computations}, the arrows $SA \rightarrow SB$ as \emph{state transformers}, and the arrows $PA \rightarrow PB$ as \emph{predicate transformers}.

We refer to $\alpha$ and $\beta$ as the \emph{validity transformations}, since the intuition is that $\alpha_A(p)(x) = \beta_A(x)(p)$ is the probability of the statement 'Predicate $p$ is valid at state $x$'.

\begin{center}
\includegraphics[trim=4.5cm 22.5cm 10cm 4cm,clip=true]{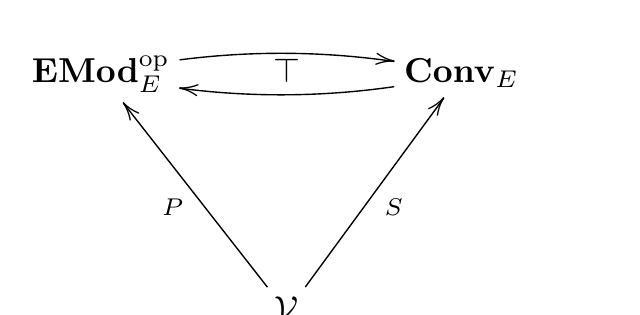}
\end{center}

\subparagraph{Examples}
The following are all examples of state-and-effect triangles:
\begin{itemize}
\item
Take $\mathcal{V}$ to be the category $\FdHilbUn$ of finite-dimensional Hilbert spaces with unitary maps, $PH$ to be the set of \emph{effects} on $H$ (positive operators less than $I$), and $SH$ to be the set of \emph{density matrices} on $H$.
\item
Take $\mathcal{V}$ to be $\Kl{\mathcal{D}_E}$, the Kleisli category of the distribution monad $\mathcal{D}_E$.  $S$ is the canonical functor from the Kleisli category to the Eilenberg-Moore category.  For $X \in Set$, $PX$ is the set of all functions $X \rightarrow \mathcal{D}_E(2)$, equivalently the set of functions $X \rightarrow E$.
\item
Take $\mathcal{V}$ to be
$\CStar$, the category of C$^*$-algebras and positive unital maps.  $PA$ is the set of all effects on $a$, $[0,1]_A = \{ a \in A : 0 \leq a \leq 1 \}$.  $SA$ is the set of all positive unital maps $A \rightarrow \mathbb{C}$.
\item
Take $\mathcal{V}$ to be $\Set$, $PA$ the power set of $A$, and $SA = A$.  The effect monoid in this case is $\{0,1\}$.
\pagebreak
\item
More generally, let $(\mathcal{V}, \otimes, I)$ be any symmetric monoidal category with finite coproducts $(0, +)$ such that:
\begin{itemize}
\item
diagrams of the following form are always pullbacks in $\mathcal{V}$:

\begin{center}
\includegraphics[trim=4.5cm 22.5cm 6cm 4cm,clip=true]{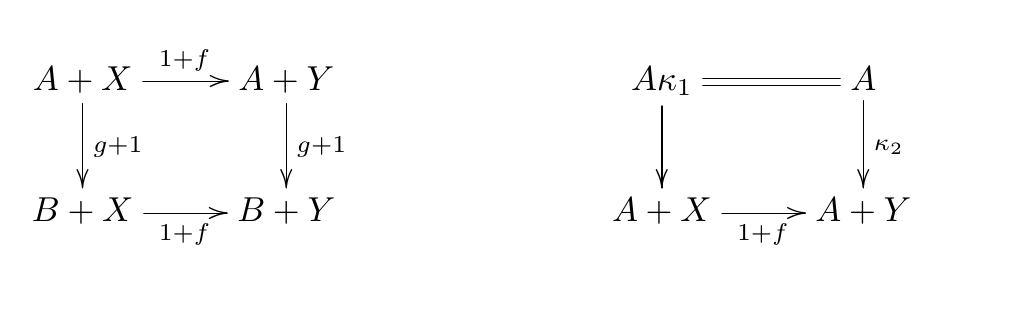}
\end{center}

\item
for each non-zero $n \in \mathbb{N}$, the family of maps
\[ [\rhd_i, \kappa_2] : n \cdot X + 1 \rightarrow X + 1 \]
are jointly monic where
where, for $1 \leq i \leq n$, the `partial projection' $\rhd_i : n \cdot X \rightarrow X + 1$ is such that
\[
\rhd_i \circ \kappa_j = \begin{cases}
\kappa_1 & \mbox{if } i = j \\
\kappa_2 \circ ! & \mbox{if } i \neq j
\end{cases}
\]
\end{itemize}

Take $E = \mathcal{V}[1,2]$.
Denife $P$ to be the functor $\mathcal{V}[-,2]$, and $S$ to be the functor $\mathcal{V}[-,1]$.  
Define $\alpha$ and $\beta$ by
\[\alpha(p)(\omega) = \beta(\omega)(p) = p \circ \omega \enspace . \]
for $p : A \rightarrow 2$ and $\omega : 1 \rightarrow A$.

The arrow $\meas_A(r_1, \ldots, r_n)$ is the unique arrow such that $\rhd_i \circ \meas_A(r_1, \ldots, r_n) = r_i$.

See \cite{Jacobs} for a verification that these constructions are all well-defined and satisfy the axioms of a state-and-effect triangle.  The previous examples are all special cases of this construction.
\end{itemize}

\subparagraph{Remarks}
\begin{enumerate}
\item
We do not want $S$ always to be a strong monoidal functor.  Intuitively, $S(A) \otimes S(B)$ gives the mixtures of pure states of $A \otimes B$, while $S(A \otimes B)$ also includes entangled states, and these will not be isomorphic in general.
\item
The condition $\alpha_A(p)(x) = \beta_A(x)(p)$ can also be written as $\alpha = G \beta \circ \eta P$ or as $\beta = F \alpha \circ \epsilon S$,
where $F = \EMod{E}[-,E] : \EMod{E}^\mathrm{op} \rightarrow \Conv{E}$ and $G = \Conv{E}[-,E] : \Conv{E} \rightarrow \EMod{E}^\mathrm{op}$.
\end{enumerate}

\subsection{Semantics}
\label{section:completeness}

\begin{df}
Given any state-and-effect triangle, we interpret the syntax as follows.
\begin{itemize}
\item
We associate with every type $A$ an object $\brackets{A}$ of $\mathcal{V}$ thus:
\begin{align*}
\brackets{I} & = I \\
\brackets{A \otimes B} & = \brackets{A} \otimes \brackets{B} \\
\brackets{A + B} & = \brackets{A} + \brackets{B}
\end{align*}
\item
We associate with every context $\Gamma$ an object $\brackets{\Gamma}$ of $\mathcal{V}$ as follows.
\begin{align*}
\brackets{\langle \rangle} & = I \\
\brackets{\Gamma, x : A} & = \brackets{\Gamma} \otimes \brackets{A}
\end{align*}
\item
We associate with every term $\Gamma \vdash M : A$ an arrow $\brackets{M} = \brackets{\Gamma \vdash M : A} : \brackets{\Gamma} \rightarrow \brackets{A}$ in $\mathcal{V}$ as follows.
\begin{itemize}
\item
$\brackets{x_1 : A_1, \ldots, x_n : A_n \vdash x_i : A_i}$ is the arrow

\begin{center}
\includegraphics[trim=5cm 22cm 3cm 4.5cm,clip=true]{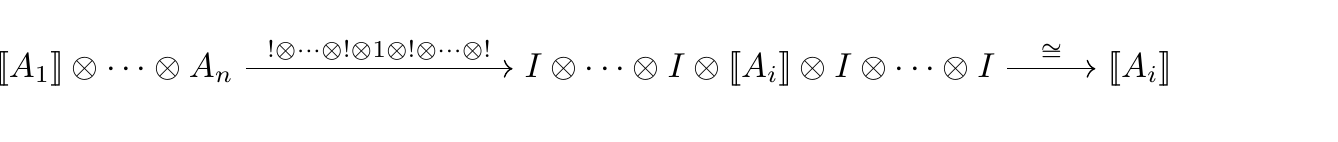}
\end{center}
\item
$\brackets{M \otimes N} = \brackets{M} \otimes \brackets{N}$
\item
$\brackets{\Gamma, \Delta \vdash \lett x \otimes y = M \inn N : C}$ is

\begin{center}
\includegraphics[trim=5cm 22.5cm 3cm 4cm,clip=true]{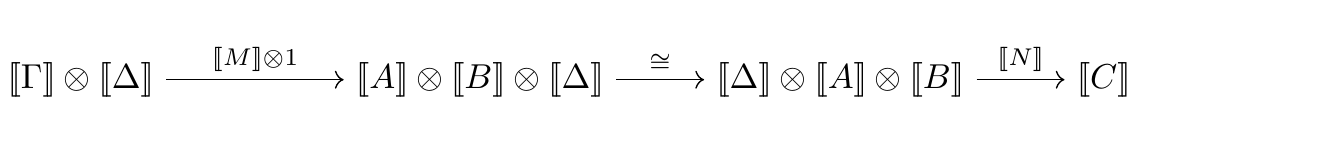}
\end{center}
\item
$\brackets{\Gamma \vdash \langle \rangle : I} = ! : \brackets{\Gamma} \rightarrow I$
\item
$\brackets{\inl{M}} = \kappa_1 \circ \brackets{M}$
\item
$\brackets{\inr{M}} = \kappa_2 \circ \brackets{M}$
\item
$\brackets{\Gamma, \Delta \vdash \case M \of \inl{x} \mapsto N \mid \inr{y} \mapsto P : C }$ is the arrow

\begin{center}
\includegraphics[trim=5cm 22cm 3cm 4.5cm,clip=true]{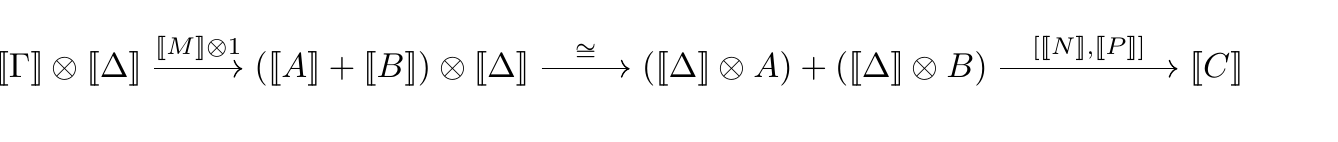}
\end{center}

\item
$\brackets{\Gamma, \Delta \vdash \measure \phi_1 \mapsto M_1 \mid \cdots \mid \phi_n \mapsto M_n : A}$
is the arrow

\begin{center}
\includegraphics[trim=5cm 22cm 3cm 4.5cm,clip=true]{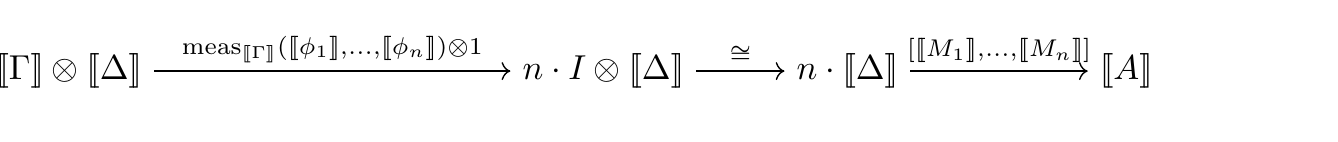}
\end{center}
\end{itemize}
\item
We associate with every proposition $\phi$ such that $\Gamma \vdash \phi \prop$, an element $\brackets{\phi} \in P \brackets{\Gamma}$ as follows.
\begin{align*}
\brackets{0} & = 0 \\
\brackets{\phi^\bot} & = \brackets{\phi}^\bot \\
\brackets{\phi \ovee \psi} & = \brackets{\phi} \ovee \brackets{\psi} \\
\brackets{\phi \cdot \psi} & = \brackets{\phi} \cdot \brackets{\psi}
\end{align*}
In this last line, if $\Gamma \vdash \phi \cdot \psi \eff$ then $\brackets{\phi} \in P I$ and $\brackets{\psi} \in P \brackets{\Gamma}$.
We use the fact that $E \cong PI$ (since $P$ preserves the terminal object), so we may take $\brackets{\phi}$ to be an element of $E$.

$\brackets{\Gamma, \Delta \vdash \case M \of \inl{x} \mapsto \phi \mid \inr{y} \mapsto \psi}$ is defined as follows.  We have
\[ 1 \otimes \brackets{M} : \brackets{\Gamma} \otimes \brackets{\Delta} \rightarrow \brackets{\Gamma} \otimes (\brackets{A} + \brackets{B}) = (\brackets{\Gamma} \otimes \brackets{A}) + (\brackets{\Gamma} \otimes \brackets{B}) \]
and so
\[ P(1 \otimes \brackets{M}) : P(\brackets{\Gamma} \otimes \brackets{A}) \times P(\brackets{\Gamma} \otimes \brackets{B}) \rightarrow
P(\brackets{\Gamma} \otimes \brackets{\Delta}) \]
(Recall that $P : \mathcal{V} \rightarrow \EMod{E}^\op$ preserves binary coproducts, and so $P(A + B)$ is the product of $PA$ and $PB$ in $\EMod{E}$.)

We define $\brackets{\case M \of \inl{x} \mapsto \phi \mid \inr{y} \mapsto \psi}$ to be
\[ P(1 \otimes \brackets{M}) (\brackets{\phi}, \brackets{\psi}) \enspace . \]
\end{itemize}
\end{df}

\begin{lemma}
\begin{enumerate}
\item
If $\Gamma, x : A \vdash M : B$ and $\Delta \vdash N : A$, then $\brackets{\Gamma, \Delta \vdash [N/x]M : B}$ is the arrow

\begin{center}
\includegraphics[trim=5cm 22cm 3cm 4.5cm,clip=true]{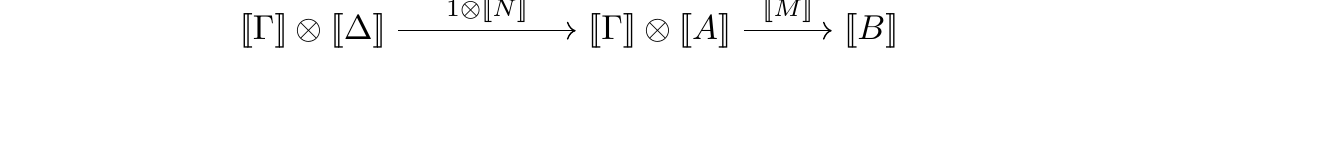}
\end{center}

\item
If $\Gamma, x : A \vdash \phi \prop$ and $\Delta \vdash M : A$ then
\[ \brackets{[M/x]\phi} = P(1_{\brackets{\Gamma}} \otimes \brackets{M})(\brackets{\phi}) \]
\end{enumerate}
\end{lemma}

\begin{proof}
The two parts are proved simultaneously, by induction on $M$ and $\phi$.  All cases are straightforward.
\end{proof}


\begin{df}
In a state-and-effect triangle,
a judgement $\Gamma \vdash M = N : A$ is \emph{true} iff $\brackets{\Gamma \vdash M : A} = \brackets{\Gamma \vdash N : A}$.
A judgement $\Gamma \vdash \phi \leq \psi$ is \emph{true} iff $\brackets{\Gamma \vdash \phi \prop} \leq \brackets{\Gamma \vdash \psi \prop}$, in the order in the effect module $P \brackets{\Gamma}$.
\end{df}

\begin{theorem}[Soundness]
Any derivable judgement is true in any state-and-effect triangle.
\end{theorem}

\begin{proof}
Straightforward induction on derivations.
\end{proof}

\begin{theorem}[Completeness]
Any judgement that is true in every state-and-effect triangle is derivable.
\end{theorem}

\begin{proof}
Define a state-and-effect triangle as follows.

The category $\mathcal{V}$ is the category with objects the types of QPEL, and arrows $A \rightarrow B$ the pairs $(x,M)$ such that $x : A \vdash M : B$, quotiented by:
\begin{itemize}
\item
$(x : A \vdash M : B) = (y : A \vdash [y/x]M : B)$ if $x \neq y$ and $y$ does not occur in $M$;
\item
If $x : A \vdash M = N : B$ is derivable, then $(x : A \vdash M : B) = (x : A \vdash N : B)$.
\end{itemize}
The identity on $A$ is $x : A \vdash x : A$.  The composite of $x : A \vdash M : B$ and $y : B \vdash N : C$ is $x : A \vdash [M/y]N : C$.  This is well-defined by Substitution and Functionality.

We shall write an arrow $x : A \vdash M : B$ as $M[x] : A \rightarrow B$, and then write $M[N]$ for the term $[N/x]M$.

\subparagraph{Tensor Product}

For types $A$ and $B$, the tensor product is $A \otimes B$.

Given arrows $M[a] : A \rightarrow A'$ and $N[b] : B \rightarrow B'$, define $M \otimes N : A \otimes B \rightarrow A' \otimes B'$ by
\[ (M \otimes N)[z] = \lett a \otimes b = z \inn M[a] \otimes N[b] \enspace . \]

\subparagraph{Coproducts}

For types $A$ and $B$,  the coproduct is $A + B$, with injections
\[ x : A \vdash \inl{x} : A + B, \qquad y : B \vdash \inr{y} : A + B \enspace . \]
Given $M[a] : A \rightarrow C$ and $N[b] : B \rightarrow C$, the mediating arrow $[M,N] : A + B \rightarrow C$ is defined by
\[ [M,N][x] = \case x \of \inl{a} \mapsto M[a] \mid \inr{b} \mapsto N[b] \enspace . \]

\subparagraph{Effect Monoid}
The effect monoid $E$ is the set of all propositions $\phi$ such that $\vdash \phi \prop$,
quotiented by: $\phi = \psi$ iff $\vdash \phi \leq \psi \mbox{ and } \vdash \psi \leq \phi$.

We have that $\phi \ovee \psi$ is defined iff $\vdash \phi \ovee \psi \prop$ (equivalently, iff $\vdash \phi \leq \psi^\bot$), in which case the partial sum is $\phi \ovee \psi$.  The zero element is $0$, and the orthocomplement of $\phi$ is $\phi^\bot$.  The product of $\phi$ and $\psi$ is $\phi \cdot \psi$.

\subparagraph{Predicate Functor}
The functor $P$ is defined by: $PA$ is the set of all pairs $(x,\phi)$ such that
$x : A \vdash \phi \prop$,
quotiented by: 
\begin{itemize}
\item
$(x,\phi) = (y,[y/x]\phi)$ if $x \not\equiv y$ and $y$ does not occur in $\phi$;
\item
$(x, \phi) = (x, \psi)$ if
$x : A \vdash \phi \equiv \psi$.
\end{itemize}
This is an effect module under $0$, $^\bot$, $\ovee$, $\cdot$.

Given $M[a] : A \rightarrow B$, then $PM : PB \rightarrow PA$ is defined by
\[ PM (b, \phi) \equiv (a, [M[a]/b] \phi) \enspace . \]

\subparagraph{State Functor}
The functor $S$ is defined by: $SA$ is the set of all terms $M$ such that
$\vdash M : A$,
quotiented by: $M = N$ iff
$\vdash M = N : A$.

We make this into a convex set by setting
\[ \phi_1 M_1 + \cdots + \phi_n M_n = \measure \phi_1 \mapsto M_1 \mid \cdots \mid \phi_n \mapsto M_n \enspace . \]

Given $M[a] : A \rightarrow B$, we define $SM : SA \rightarrow SB$ by
\[ SM(N) \equiv M[N] \enspace . \]

We make $S$ into a symmetric monoidal functor by setting
\begin{align*}
\phi_{AB} & : SA \otimes SB \rightarrow S(A \otimes B) \\
\phi_{AB}(M,N) & = M \otimes N \\
\phi & : \{ * \} \rightarrow S I \\
\phi(*) & = \langle \rangle
\end{align*}

\subparagraph{Measurement Morphisms}
We have $\meas_A(\phi_1, \ldots, \phi_n) = \mathsf{measure}\ \phi_1 \mapsto \mathsf{in}_1(\langle \rangle) \mid \cdots \mid \phi_n \mapsto \mathsf{in}_n(\langle \rangle)$,
where the terms $\mathsf{in}_i(M)$ are the $n$ canonical terms such that $x : A \vdash \mathsf{in}_i(x) : \overbrace{A + \cdots + A}^n$.

\subparagraph{Validity Transformations}
The transformation $\alpha$ is given by $\alpha_A (x : A \vdash \phi \prop) (\vdash M : A) \equiv (\vdash [M / x] \phi \prop)$, and so
$\beta$ is given by
$\beta_A (\vdash M : A) (x : A \vdash \phi \prop) \equiv (\vdash [M / x] \phi \prop)$.

\subparagraph{Proof of Completeness}
We will prove that, if a judgement is true in this triangle, then it is derivable.

Let $\Gamma \equiv x_1 : A_1, \ldots, x_n : A_n$.  Then
a straightforward induction shows that:
\begin{align*}
\brackets{\Gamma \vdash M : B} & =
z : A_1 \otimes \cdots \otimes A_n \vdash \lett x_1 \otimes \cdots \otimes x_n = z \inn M : B \\
\brackets{\Gamma \vdash \phi \eff} & =
z : A_1 \otimes \cdots \otimes A_n \vdash \lett x_1 \otimes \cdots \otimes x_n = z \inn \phi \eff
\end{align*}
where this last effect is defined inductively thus:
\begin{align*}
\lett x_1 \otimes \cdots \otimes x_n = z \inn 0 & \equiv 0 \\
\lett x_1 \otimes \cdots \otimes x_n = z \inn \phi^\bot & \equiv
(\lett x_1 \otimes \cdots \otimes x_n = z \inn \phi)^\bot \\
\lett x_1 \otimes \cdots \otimes x_n = z \inn \phi \ovee \psi & \equiv
(\lett x_1 \otimes \cdots \otimes x_n = z \inn \phi) \ovee (\lett x_1 \otimes \cdots \otimes x_n = z \inn \psi) \\
\begin{array}{r}
\lett x_1 \otimes \cdots \otimes x_n = z \inn \\
\case M \of \inl{x} \mapsto \phi \mid \inr{y} \mapsto \psi \end{array} &
\equiv \begin{cases}
\case (\lett x_1 \otimes \cdots \otimes x_n = z \inn M) \of \\
\qquad \inl{x} \mapsto \phi \mid \inr{y} \mapsto \psi \\
\qquad \mbox{if } z \mbox{ occurs in } M \\
\case M \of \inl{x} \mapsto \lett x_1 \otimes \cdots \otimes x_n = z \inn \phi
\mid\\
\qquad  \inr{y} \mapsto \lett x_1 \otimes \cdots \otimes x_n = z \inn \psi
\\
\qquad  \mbox{otherwise}
\end{cases}
\end{align*}

Suppose that the judgement $\Gamma \vdash M = N : A$ is true in this triangle.  Then we have
\[ z : A_1 \otimes \cdots \otimes A_n \vdash (\lett x_1 \otimes \cdots \otimes x_n = z \inn M) =
(\lett x_1 \otimes \cdots \otimes x_n = z \inn N) : A \]
is derivable.
By Substitution, we have
\[ \Gamma \vdash (\lett x_1 \otimes \cdots \otimes x_n = x_1 \otimes \cdots \otimes x_n \inn M) =
(\lett x_1 \otimes \cdots \otimes x_n = x_1 \otimes \cdots \otimes x_n \inn N) : A \]
is derivable, and hence $\Gamma \vdash M = N : A$ is derivable by \Retaotimes.

Suppose that $\Gamma \vdash \phi \leq \psi$ is true in this triangle.  Then
\[ z : A_1 \otimes \cdots \otimes A_n \vdash (\lett x_1 \otimes \cdots \otimes x_n = z \inn \phi) \leq (\lett x_1 \otimes \cdots \otimes x_n = z \inn \psi) \]
is derivable.  By Substitution, we have
\[ \Gamma \vdash (\lett x_1 \otimes \cdots \otimes x_n = x_1 \otimes \cdots \otimes x_n \inn \phi) \leq
(\lett x_1 \otimes \cdots \otimes x_n = x_1 \otimes \cdots \otimes x_n \inn \phi) \enspace . \]
It is easy to show, by induction on $\phi$, that
\[ \Gamma \vdash ((\lett x_1 \otimes \cdots \otimes x_n = x_1 \otimes \cdots \otimes x_n \inn \phi) \equiv \phi \enspace . \]
It follows that $\Gamma \vdash \phi \leq \psi$ is derivable.
\end{proof}

\section{Qubits}
\label{section:qubits}

There are several ways in which the system may be extended to represent qubits.  The details below are based on the Measurement Calculus \cite{Danos2009}.

We extend the system with:
\[ \begin{array}{lrcl}
\mbox{Type} & A & ::= & \cdots \mid \qbit \\
\mbox{Term} & M & ::= & \cdots \mid \ket{+} \mid XM \mid ZM \mid EMM \\
\mbox{Effect} & \phi & ::= & \cdots \mid M = \ket{+_\alpha}
\end{array} \]
where $\alpha$ is a real number in $[0, 2 \pi)$.

The intention is that a term of type $\qbit$ represents a qubit.  The term $\ket{+}$ represents a qubit in the phase
\[ \ket{+} = \frac{1}{\sqrt{2}}(\ket{0} + \ket{1}) \enspace . \]
The terms $XM$ and $ZM$ denote the result of applying the Pauli-X and Z gates to the qubit $M$.  The term $EMN$ denotes the result of applying the controlled Z gate to the pair of qubits $M$ and $N$.  The effect $M = \ket{+_\alpha}$ denotes the projector on

\[ \ket{+_\alpha} = \frac{1}{\sqrt{2}}(\ket{0} + e^{i \alpha} \ket{1}) \]

Its orthocomplement, $\ket{+_\alpha}^\bot$, is the projector on

\[ \ket{-_\alpha} = \frac{1}{\sqrt{2}}(\ket{0} - e^{i \alpha} \ket{1}) \]

We write $\ket{-}$ for $Z \ket{+}$

We extend the system with the following rules of deduction.

\begin{center}
\AxiomC{}
\UnaryInfC{$\vdash \new \ket{+} : \qbit$}
\DisplayProof
\qquad
\AxiomC{$\Gamma \vdash M : \qbit$}
\UnaryInfC{$\Gamma \vdash XM : \qbit$}
\DisplayProof
\qquad
\AxiomC{$\Gamma \vdash M : \qbit$}
\UnaryInfC{$\Gamma \vdash ZM : \qbit$}
\DisplayProof
\end{center}

\begin{center}
\AxiomC{$\Gamma \vdash M : \qbit$}
\AxiomC{$\Gamma \vdash N : \qbit$}
\BinaryInfC{$\Gamma \vdash EMN : \qbit \otimes \qbit$}
\DisplayProof
\qquad
\AxiomC{$\Gamma \vdash M : \qbit$}
\RightLabel{$(0 \leq \alpha < 2 \pi$)}
\UnaryInfC{$\Gamma \vdash M = \ket{+_\alpha} \eff$}
\DisplayProof
\end{center}

\begin{center}
\AxiomC{$\Gamma \vdash M : \qbit$}
\AxiomC{$\Gamma \vdash N : \qbit$}
\BinaryInfC{$\Gamma \vdash E(XM)N = \lett x \otimes y = EMN \inn Xx \otimes Zy : \qbit \otimes \qbit$}
\DisplayProof
\qquad
\AxiomC{$\Gamma \vdash M : \qbit$}
\AxiomC{$\Gamma \vdash N : \qbit$}
\BinaryInfC{$\Gamma \vdash E(ZM)N = \lett x \otimes y = EMN \inn Zx \otimes y : \qbit \otimes \qbit$}
\DisplayProof
\end{center}

\begin{center}
\AxiomC{$\Gamma \vdash M : \qbit$}
\RightLabel{$(0 \leq \alpha < 2 \pi$)}
\UnaryInfC{$\Gamma \vdash (XM = \ket{+_\alpha}) \equiv (M = \ket{+_{- \alpha}})$}
\DisplayProof
\qquad
\AxiomC{$\Gamma \vdash M : \qbit$}
\RightLabel{($0 \leq \alpha < 2 \pi$)}
\UnaryInfC{$\Gamma \vdash (ZM = \ket{+_\alpha}) \equiv (M = \ket{+_{\alpha - \pi}})$}
\DisplayProof
\end{center}

\begin{center}
\AxiomC{$\Gamma \vdash M : \qbit$}
\UnaryInfC{$\Gamma \vdash X(XM) = M : \qbit$}
\DisplayProof
\qquad
\AxiomC{$\Gamma \vdash M : \qbit$}
\UnaryInfC{$\Gamma \vdash Z(ZM) = M : \qbit$}
\DisplayProof
\end{center}

\begin{center}
\AxiomC{$\Gamma \vdash M : \qbit$}
\RightLabel{$(0 \leq \alpha < 2 \pi$)}
\UnaryInfC{$(X(ZM) = \ket{+_\alpha}) \equiv (Z(XM) = \ket{+_\alpha})$}
\DisplayProof
\end{center}

The metatheorems in Section \ref{section:metatheorems} all still hold for the expanded system.
The expanded system can be given semantics in $\CStar$ straightforwardly.  We will show in a forthcoming paper how these rules are sufficient to prove the correctness of several quantum algorithms, including superdense coding and gate-based teleportation.

\section{Natural Isomorphisms}

It is interesting to consider the question of when the natural transformations $\alpha$ and $\beta$ are isomorphisms.
In the $\FdHilbUn$ example, $\alpha$ and $\beta$ are both isomorphisms. \cite{Jacobs2013}.
In the $\Kl{\mathcal{D}}$ example, $\alpha$ is an isomorphism but $\beta$ is not.
In the $\CStar$ example, $\beta$ is an isomorphism but $\alpha$ is not.

We can extend the system so it captures the state-and-effect triangles in which $\beta$ is an isomorphism as follows.

\begin{theorem}[Completeness]
Add to the system the rule
\begin{prooftree}
\AxiomC{$\vdash \phi \prop$}
\AxiomC{$\Gamma \vdash M : A$}
\AxiomC{$\Gamma \vdash N : A$}
\AxiomC{$\Delta, x : A \vdash \psi \prop$}
\QuaternaryInfC{$\Gamma, \Delta \vdash [(\mathsf{measure}\ \phi \mapsto M | \phi^\bot \mapsto N)/x]\psi \equiv
(\phi \cdot [M/x]\psi) \ovee (\phi^\bot \cdot [N/x]\psi)$}
\end{prooftree}
If a judgement is true in every state-and-effect triangle in which $\alpha$ and $\beta$ are natural isomorphisms, then it is derivable in this system.
\end{theorem}

I do not yet have a system that captures the state-and-effect triangles in which $\alpha$ is a natural isomorphism.  

The case where $\alpha$ is an isomorphism is particularly interesting, as it is this that allows \emph{weakest preconditions} in d'Hondt-Panangaden's sense to be defined.  

\begin{df}
Let $P$ and $Q$ be quantum predicates, and $F$ a quantum program.
Then $P$ is a \emph{precondition} for $Q$ with respect to $M$, ${P} F {Q}$, iff for all density matrices $\rho$, $\trace(P \rho) \leq \trace(Q F(\rho))$.
$P$ is the \emph{weakest precondition} for $Q$ with respect to $M$,
$P = wp(F)(Q)$
iff $P$ is the greatest precondition for $Q$ w.r.t.~$M$ under the L\"owner order.
\end{df}

The weakest precondition for $Q$ w.r.t.~$F$ always exists and is unique \cite{2006}.

\begin{lemma}
In the $\FdHilbUn$ state-and-effect triangle, the weakest precondition for $Q \in PH$ with respect to $F : SK \rightarrow SH$ is
$\alpha^{-1}(F \circ \alpha(P))$.
The operation $wp(F)$ is therefore the effect module homomorphism
$\alpha^{-1} \circ \Conv{M}[1, F] \circ \alpha : P H \rightarrow P K$.
The operation $wp$ is therefore the natural transformation
\[ wp_{HK} = \alpha^{-1} \circ \Conv{M}[1, -] \circ \alpha : \Conv{M}[SK, SH] \rightarrow \EMod{M}[PH, PK] \]
\end{lemma}

\begin{lemma}
Given $\Gamma \vdash M : A$ and $x : A \vdash \phi \prop$, then in the $\FdHilbUn$ semantics:
\[wp(\brackets{\Gamma \vdash M : A})(\brackets{x : A \vdash \phi \prop}) = \brackets{\Gamma \vdash [M/x]\phi \prop} \]
\end{lemma}

\section{Conclusion, Related Work and Future Work}

We have presented QPEL, a syntactic system involving both terms and propositions that captures the categorical notion of `state-and-effect triangle' which has proved to be a general setting for describing both quantum programs, and effects.  It is therefore a promising candidate for a language that allows us to reason about and prove properties of quantum programs, and shows how such a logic for quantum effects might be added on top of any quantum programming language.

Baltag and Smets in a series of papers \cite{Baltag,Baltag2012,Baltag2011,Baltaga,Baltag2004} describe the language QDL, Quantum Dynamic Logic.  This is also a language for describing quantum programs and properties of quantum programs.  Their work differs from mine because their term language is an underspecification language (as is Dynamic Logic's), and their propositions can denote all propositions expressible in classical logic, not just those that correspond to quantum effects.

d'Hondt-Panangaden \cite{2006} and Ying \cite{Ying2011} have investigated the notion of a \emph{quantum predicate}.  Ying has given a Floyd-Hoare style logic which, given a program $F$ written in his syntax, allows the weakest precondition of a predicate with respect to $F$ to be calculated.  Their work differs from mine because they do not give a syntax for the predicates, instead using the effects on a Hilbert space as the predicates directly.

In the future, the most important tasks are to apply the system to prove the correctness of a simple quantum program (e.g. the quantum teleportation protocol or quantum broadcasting), and to look for ways to extend the system in order to represent looping and/or recursion.

I will present the system in a more modular fashion, giving subsystems that can be interpreted in other state-and-effect triangles, for example using complete lattices instead of effect modules.  This may lead to a general notion of a (2-)category of state-and-effect triangles.

I will also try to capture the conditions that make $\alpha$ or $\beta$ a natural isomorphism.  I will investigate the conditions that a state-and-effect triangle needs to satisfy to represent the type of qubits correctly, possibly involving Selinger's notion of a Quantum Flowchart Category.  I will investigate formal translations between this system and other quantum programming languages, such as the quantum lambda calculus \cite{Sellinger2010}.  I will investigate which of Ying's equations on weakest preconditions \cite{Ying2011} can be derived within our system.


\paragraph{Acknowlegdements}

Thanks to Sam Staton and Bart Jacobs for many helpful discussions.

\bibliography{type}
\end{document}